\documentclass[10pt,conference]{IEEEtran}

\usepackage{graphicx}
\usepackage{amsmath}
\usepackage{amssymb}
\usepackage{multirow}

\usepackage{cite}
\usepackage[utf8]{inputenc}
        \usepackage[T1]{fontenc}
\usepackage{amsthm}
\usepackage{algorithm}
\usepackage{algpseudocode}

\newtheorem*{proposition*}{Proposition}
\usepackage{hyperref}
\usepackage{cleveref}
\begin{document}
\newcommand{\mat}{\mathbf}
\renewcommand{\vec}[1]{\bar{\mat{#1}}}

\title{On Optimal Offline Time Sharing Policy for Energy Harvesting Underlay Cognitive Radio}
\author{\IEEEauthorblockN{Kalpant~Pathak and Adrish~Banerjee}
\IEEEauthorblockA{Department of Electrical Engineering, Indian Institute of Technology Kanpur, Uttar Pradesh, 208016\\
Email: \{kalpant, adrish\}@iitk.ac.in}
\vspace{0mm}
}


\maketitle

\begin{abstract}
RF energy harvesting can be used to power communication devices so that perpetual operation of such devices can be ensured. We consider a RF energy harvesting underlay cognitive radio system operating in slotted fashion. The primary user (PU) is equipped with a reliable power source and transmits with a constant power in all the slots. However, the secondary user (SU) harvests energy from primary's transmission and simultaneously transmits it's own data such that interference at the primary receiver (PR) remains below an acceptable threshold. At the secondary transmitter (ST), each time slot is divided into two phases: energy harvesting (EH) phase and information transfer (IT) phase. We formulated the problem of maximizing the achievable secondary sum rate under primary receiver's protection criteria as a convex optimization problem and obtained the optimal time sharing between EH and IT phase, and optimal secondary transmit power under offline setting. The optimal offline scheme is then compared with an online myopic policy, where the optimal time sharing between EH and IT phase, and optimal secondary transmit power are obtained based on instantaneous channel gains only.
\end{abstract}

\IEEEpeerreviewmaketitle


\section{Introduction}
In recent years, energy harvesting cognitive radio networks (EH-CRNs) has emerged as a solution to the problem of spectrum scarcity and at the same time ensures perpetual operation of the communication devices \cite{EH_CRN}. In the literature, the EH-CRNs operating in interweave, overlay and underlay mode have been studied in  \cite{interweave_1,interweave_2,interweave_3,interweave_4,interweave_5}, \cite{overlay_1,overlay_2,
overlay_3}, and \cite{overlay_4,underlay_1,underlay_2,underlay_3}, respectively.

In \cite{interweave_1}, authors considered an EH-CR network operating in interweave mode. The secondary user (SU) uses ambient radio signals and wireless power transfer to harvest RF energy. The authors used time homogeneous discrete Markov process to model the primary traffic. In each slot, the SU decides either to remain idle or to perform spectrum sensing and re-configures the detection threshold. The authors proposed a transmission policy such that expected total throughput of SU is maximized. They formulated the optimization problem as constrained partially observable Markov decision process (POMDP) and obtained a policy maximizing SU's throughput by designing the spectrum sensing policy and detection threshold jointly. In \cite{interweave_2}, authors considered the system model of \cite{interweave_1}, and used time homogeneous discrete Markov process to model the temporal correlation of the primary traffic. The authors upper bounded the achievable throughput of the SU, which is based on energy arrival rates, temporal correlation of primary traffic and detection threshold. The authors obtained an optimal detection threshold maximizing the upper bound on achievable throughput. In \cite{interweave_3}, authors considered a CR network with energy harvesting secondary user in a single user multi-channel scenario. Based on energy availability of SU, channel conditions and belief state of PU's network, authors obtained a channel selection criterion. This channel selection criterion chooses the best subset among all the available channels for sensing. Then, using the proposed criterion, the authors constructed a channel-aware optimal and myopic sensing policy. In \cite{interweave_4}, authors considered an EH-CR network where both the PU and SU have energy harvesting capability but they don't have any battery to store the harvested energy. So, if harvested energy is not used, it is discarded. Based on Markovian behavior of PU, to specify the interaction between the PUs and SUs in the system, authors adapted a hidden input Markov model. The authors proposed a two-dimensional spectrum and power sensing policy that improves the PU detection performance and estimates the primary transmit power level. In \cite{interweave_5}, authors obtained an optimal detection threshold specifying a sensing policy that maximizes the expected total throughput of energy harvesting SU. Authors showed that when arrival rate of energy is less than it's expected consumption, decreased probability of accessing occupied channel may not affect the probability of accessing the vacant spectrum.

\indent In \cite{overlay_1}, authors considered an EH-CR network operating in overlay scenario where the secondary user helps primary deliver it's data. The SU has energy harvesting capability and the authors used a discrete time queue to model the energy queue with Markov arrival and service processes. In the system model considered, when PU transmits, SU remains silent and receives some fraction of PU's data and stores it in a queue. Then, as PU becomes silent, SU relays primary data using decode and forward protocol. For the proposed system, the authors obtained inner and outer bounds on the stability region. In \cite{overlay_2}, authors considered a CR network with joint information and energy cooperation. In the system model considered, the PU not only gives information to SU for relaying but feeds it with energy also. For such cooperation, authors proposed three schemes. In first, they considered an ideal backhaul for information and energy transfer between the two systems. Then, authors proposed power and time splitting schemes enabling the joint information-energy cooperation. Finally, authors obtained zero forcing solutions for all three schemes. In \cite{overlay_3}, authors considered joint information and energy cooperation among the primary and secondary users. The SU relays primary data and gets spectrum access as a reward. Moreover, the PU feeds the SU with energy. The authors aim to maximize the primary rate subject to rate constraints of both the users. Then, authors analyzed the effects of SU rate constraint and finite battery on probability of cooperation and primary rate. In \cite{overlay_4}, authors considered a scenario where PUs harvest energy from multiple access points (APs) while SUs have fixed power supply. In the \emph{energy harvesting zone} centered at PU, nearest ST is selected to transfer energy to PU wirelessly. Also, there exist a \emph{cooperative region} between PU and AP, where an ST is selected to relay primary data based on the channel quality between ST and AP, and in return, the ST is rewarded with some fraction of bandwidth of primary channel. Under performance constraints of primary system, authors aimed to maximize the throughput of secondary systems in a given area.
\begin{figure*}[!t]
\includegraphics[width=0.8\linewidth]{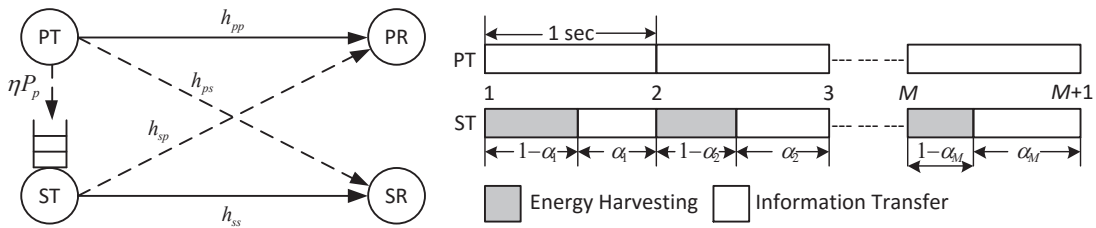}
\centering
\caption{Primary and RF energy harvesting secondary users coexisting in an underlay scenario}
\label{fig:system_model}
\vspace{-3mm}
\end{figure*}

\indent In \cite{underlay_1}, authors considered an underlay CR system where multiple PUs and energy harvesting SUs coexist. Each primary transmitter (PT) has two concentric zones, namely, \emph{guard zone} and \emph{harvesting zone}. If a secondary transmitter (ST) is located within the \emph{harvesting zone}, it harvests energy from primary transmission, if it is located outside the \emph{guard zone}, it transmits with fixed power, and otherwise remains idle. The authors considered independent homogeneous Poisson point processes (HPPP) to model PTs and STs, and maximized the spatial throughput of secondary network. In \cite{underlay_2}, authors solved the secondary system throughput maximization problem using geometric waterfilling with peak power constraints (GWFPP). It is showed that the proposed algorithm outperforms the primal-dual interior point method (PD-IPM). In \cite{underlay_3}, authors considered an underlay CR network where the SU harvests energy from PU's transmission. In each slot, some fraction of time is reserved for energy harvesting and remaining is reserved for information transfer. However, the authors assumed that the SU uses all of the harvested energy in that slot itself, which makes the optimization problem independent among the slots. The authors obtained an online myopic solution for  time sharing between the energy harvesting phase and information transfer phase which maximizes the achievable rate of CR system.

\indent In this paper, we consider an underlay CR network with one PU and one SU transmitter-receiver pair operating in slotted fashion. The PT transmits with a fixed power whereas in each slot, the ST spends some fraction of the slot for EH and remaining for IT. However in each slot, SU is not bound to use all of the energy harvested in that slot, rather it can save the energy for future use. The main contributions of this paper are as follows:
\begin{itemize}
\item First, we obtained optimal offline time sharing between EH and IT phase, and secondary transmit power policy maximizing the sum achievable throughput of ST under energy causality constraint of ST and interference constraint at primary receiver (PR). We formulate the optimization problem as convex optimization problem and obtained the solution in closed form.
\item Second, we compare the optimal offline policy with an online myopic policy where unlike the offline policy, only instantaneous channel state information (CSI) is available at ST \cite{underlay_3}. We formulate the optimization problem as convex optimization problem and decouple it into multiple parallel subproblems and obtained the optimal time sharing parameter in closed form.
\end{itemize}
The rest of the paper is organized as follows. We present the system model and problem formulation in section II. The optimal offline and the online myopic time-sharing and transmit power policies are obtained in section III, results are given in section IV and finally, we conclude in section V.

\emph{Notation}: A bold-faced symbol with a ``bar'' (e.g. $\vec{v}$) represents a vector of length $M$. $[x]^+$ represents $\max(x,0)$, and $\vec{v}\succeq0$ means that each component of vector $\vec{v}$, $v_i$ is greater than or equal to 0. The calligraphic symbols (e.g. $\mathcal{S}=\{\vec{s}_1,\ldots\vec{s}_N\}$) represents a set of $N$ vectors. And, $\vec{0}$ and $\vec{1}$ represent the vectors of all zeros and ones respectively.

\section{System Model}
We consider a scenario where a primary user (PU) and a RF energy harvesting secondary user (SU) coexist in an underlay mode as shown in Fig. \ref{fig:system_model}. Both the primary and secondary transmitters operate in slotted fashion. The PU has a reliable power supply whereas the SU harvests RF energy from PU's transmission and stores it in an infinite sized battery. We assume that the PU transmits with a constant power of $P_p$ in each slot. However in each slot, the secondary transmitter (ST) uses some fraction of the slot for harvesting energy from primary transmission and transmits in the remaining amount of time such that the interference at primary receiver remains below an acceptable threshold $P_{int}$. The objective is to maximize the sum achievable throughput of ST by the end of $M$ slots under energy causality constraints of SU and interference constraint of PU.\\
\begin{figure*}[!t]
\begin{align}
E_s^{i*}=\alpha_i^*\left[\frac{1}{\ln2(\sum_{j=i}^M\lambda_j^*+\gamma_i^*h_{sp}^i)}-\frac{1}{\theta_i}\right]^+\text{and }
\alpha_i^*=\theta_iE_s^{i*}\left[\frac{1}{\ln2\left(\log_2(1+\theta_i\beta_i^*)-\sum_{j=i}^M\eta P_p\lambda_j^*+\gamma_i^*P_{int}-\mu_i^*\right)}-1\right]^+\label{eq:opt_sol}
\end{align}
where $\beta_i^*=\left[\frac{1}{\ln2(\sum_{j=i}^M\lambda_j^*+\gamma_i^*h_{sp}^i)}-\frac{1}{\theta_i}\right]^+$ and $\theta_i=\frac{h_{ss}^i}{\sigma_s^2+h_{ps}^iP_p}$.\\
\noindent\rule{\linewidth}{0.4pt}
\vspace{-.8cm}
\end{figure*}
\indent In the system model, all the channel links are assumed to be i.i.d. Rayleigh faded with variances $\sigma_{pp}^2,\sigma_{ps}^2,\sigma_{sp}^2$ and $\sigma_{ss}^2$ for PT-PR, PT-SR, ST-PR and ST-SR link respectively, so the channel power gains $h_{pp},h_{ps},h_{sp}$ and $h_{ss}$ are i.i.d. exponentially distributed. Since we have considered an offline policy, we assume that complete channel state information (CSI) is known at ST non-causally as in \cite{non_causal_csi}. The noise at both the receivers is assumed to be zero mean additive white Gaussian with variances $\sigma_p^2$ and $\sigma_s^2$ at PR and SR respectively. Let $\alpha_i$ be the time sharing parameter between EH and IT phase such that in slot $i$, the ST harvests for $(1-\alpha_i)$ fraction of time and transmits for $\alpha_i$ fraction of time. The slot length is assumed to be 1 second without loss of generality. Let $P_s^i$ be the transmit power of ST in $i$th slot, the achievable rate of SU in bits/seconds/Hz is given by Shannon capacity formula:
\begin{align*}
R_i(\alpha_i,P_s^i)=\alpha_i\log_2\left(1+\frac{h_{ss}^iP_s^i}{\sigma_s^2+h_{ps}^iP_p}\right),\; i=1,\ldots,M.
\end{align*}
The optimization problem is given as:
\begin{subequations}
\begin{align}
\max_{\vec{\pmb{\alpha}},\vec{P}_s} \quad & \sum_{i=1}^M \alpha_i\log_2\left(1+\frac{h_{ss}^iP_s^i}{\sigma_s^2+h_{ps}^iP_p}\right)\label{eq:opt_1}\\
\text{s.t.}\quad & \sum_{j=1}^i\alpha_jP_s^j\leq \sum_{j=1}^i(1-\alpha_j)\eta P_p,\quad i=1,\ldots,M \label{eq:opt_2}\\
& \qquad\qquad \text{(Energy causality constraint of SU)} \nonumber\\
& h_{sp}^iP_s^i\leq P_{int}, \quad i=1,\ldots,M \label{eq:opt_3}\\
&\qquad\qquad\qquad\text{(Interference constraint of PU)} \nonumber\\
& \vec{P}_s\succeq \vec{0},\,\vec{0}\preceq \vec{\pmb{\alpha}}\preceq \vec{1} \label{eq:opt_4}
\end{align}
\end{subequations}
where $0\leq\eta\leq1$ and $P_{int}$ are the energy harvesting efficiency of SU and interference constraint of PR respectively. The constraint (\ref{eq:opt_2}) is the energy causality constraint of ST, which states that the total energy consumed by the end of slot $i$ must be less than or equal to total energy harvested upto slot $i$.
\section{Optimal Time Sharing and Transmit Policy}
\indent The problem (\ref{eq:opt_1})-(\ref{eq:opt_4}) is not a convex optimization problem as the optimization variables $\vec{P}$ and $\vec{\pmb{\alpha}}$ appear in product form. Let $E_s^i$ denotes the energy consumed by ST in $i$th slot so that $E_s^i=\alpha_iP_s^i$. Using change of variable $P_s^i=\frac{E_s^i}{\alpha_i}$, the optimization problem can be rewritten as:
\begin{subequations}
\begin{align}
\max_{\vec{\pmb{\alpha}},\vec{E}_s} \quad & \sum_{i=1}^M \alpha_i\log_2\left(1+\frac{h_{ss}^iE_s^i}{\alpha_i(\sigma_s^2+h_{ps}^iP_p)}\right)\label{eq:convex_opt_1}\\
\text{s.t.}\quad & \sum_{j=1}^iE_s^j\leq \sum_{j=1}^i(1-\alpha_j)\eta P_p,\quad i=1,\ldots,M \label{eq:convex_opt_2}\\
& \qquad\qquad \text{(Energy causality constraint of SU)} \nonumber\\
& h_{sp}^iE_s^i\leq \alpha_i P_{int}, \quad i=1,\ldots,M \label{eq:convex_opt_3}\\
&\qquad\qquad\qquad\text{(Interference constraint of PU)} \nonumber\\
& \vec{E}_s\succeq \vec{0},\,\vec{0}\preceq \vec{\pmb{\alpha}}\preceq \vec{1} \label{eq:convex_opt_4}
\end{align}
\end{subequations}
which is a convex optimization problem as the objective function is negative of sum of relative entropies $D(p_i||q_i)$ where $p_i=\alpha_i$ and $q_i=\alpha_i+\frac{h_{ss}^iE_s^i}{\sigma_s^2+h_{ps}^iP_p}$, and the constraints are affine inequalities. Hence, it can be solved efficiently using CVX. Firstly, we propose the necessary conditions our optimal policy must satisfy and then, we will obtain the optimal offline solution.
\subsection{Optimality Conditions}
\begin{proposition*}
The optimal time sharing and transmit power policy must satisfy $\sum_{i=1}^ME_s^{i*}= \sum_{i=1}^M(1-\alpha^*_i)\eta P_p$, i.e., the optimal policy must use all the harvested energy by end of the transmission.
\end{proposition*}
\begin{proof}
We prove this using contradiction. Let $\{\vec{E}'_s,\vec{\pmb{\alpha}}'\}$ be an optimal policy in which we have some residual energy remaining in the battery by the end of transmission, i.e., $\sum_{j=1}^ME_s^{'j}<\sum_{j=1}^M(1-\alpha'_j)\eta P_p$. Since the objective is a concave and monotonically increasing function of consumed energy, we could have consumed more energy in previous slots without violating the energy causality constraint. This way, we can also increase the transmission power. This would result in higher achievable rate, hence contradicts with our consideration of optimality.
\end{proof}

\begin{figure*}[!th]
\centering
\begin{minipage}{0.3\linewidth}
\centering
  \includegraphics[width=2.3in]{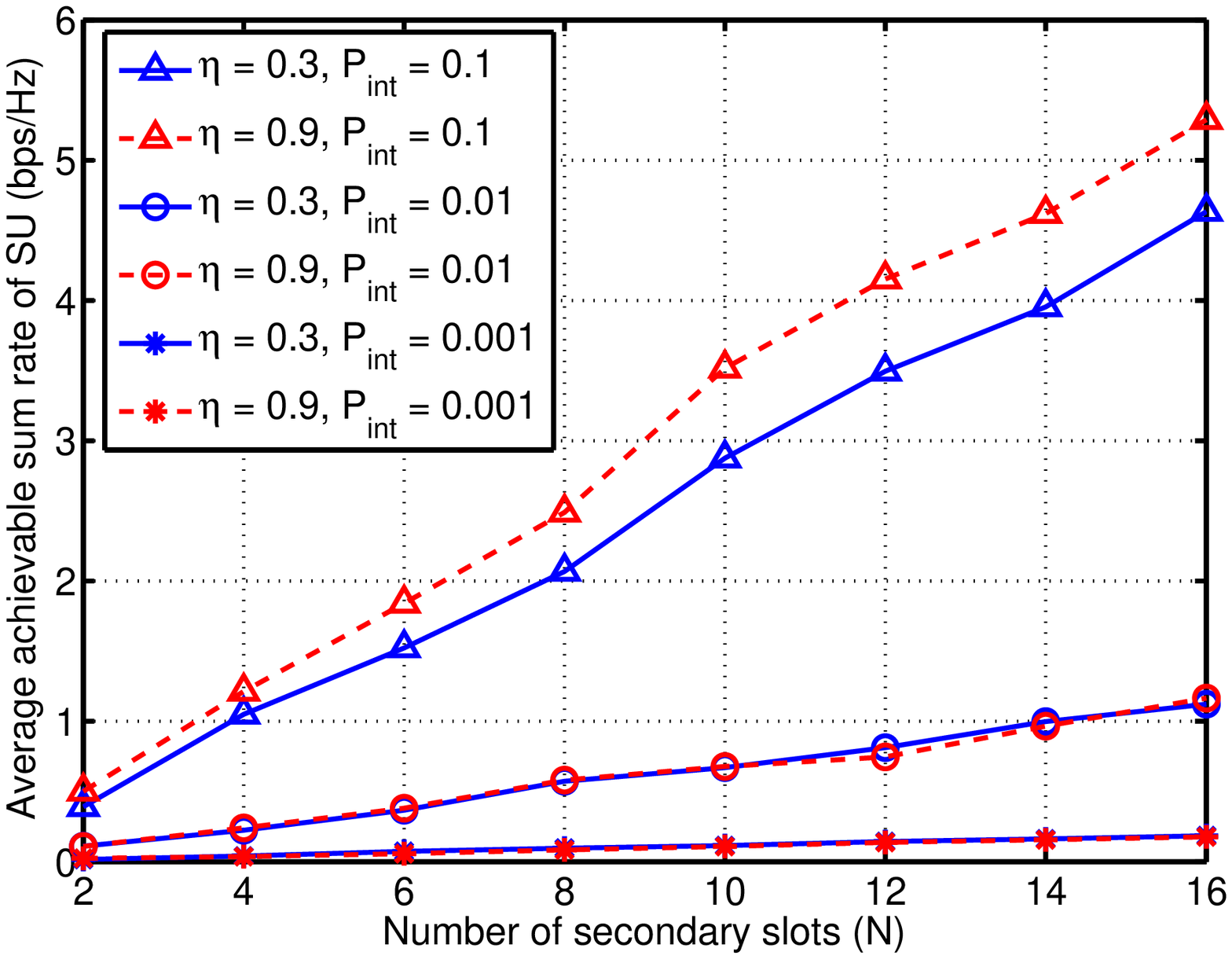}
  \caption{Average achievable sum rate of ST ($\vec{R}$) versus number of secondary slots ($N$).}
  \label{fig:opt_plot}
\end{minipage}
\hspace{3mm}
\begin{minipage}{0.3\linewidth}
\centering
  \includegraphics[width=2.3in]{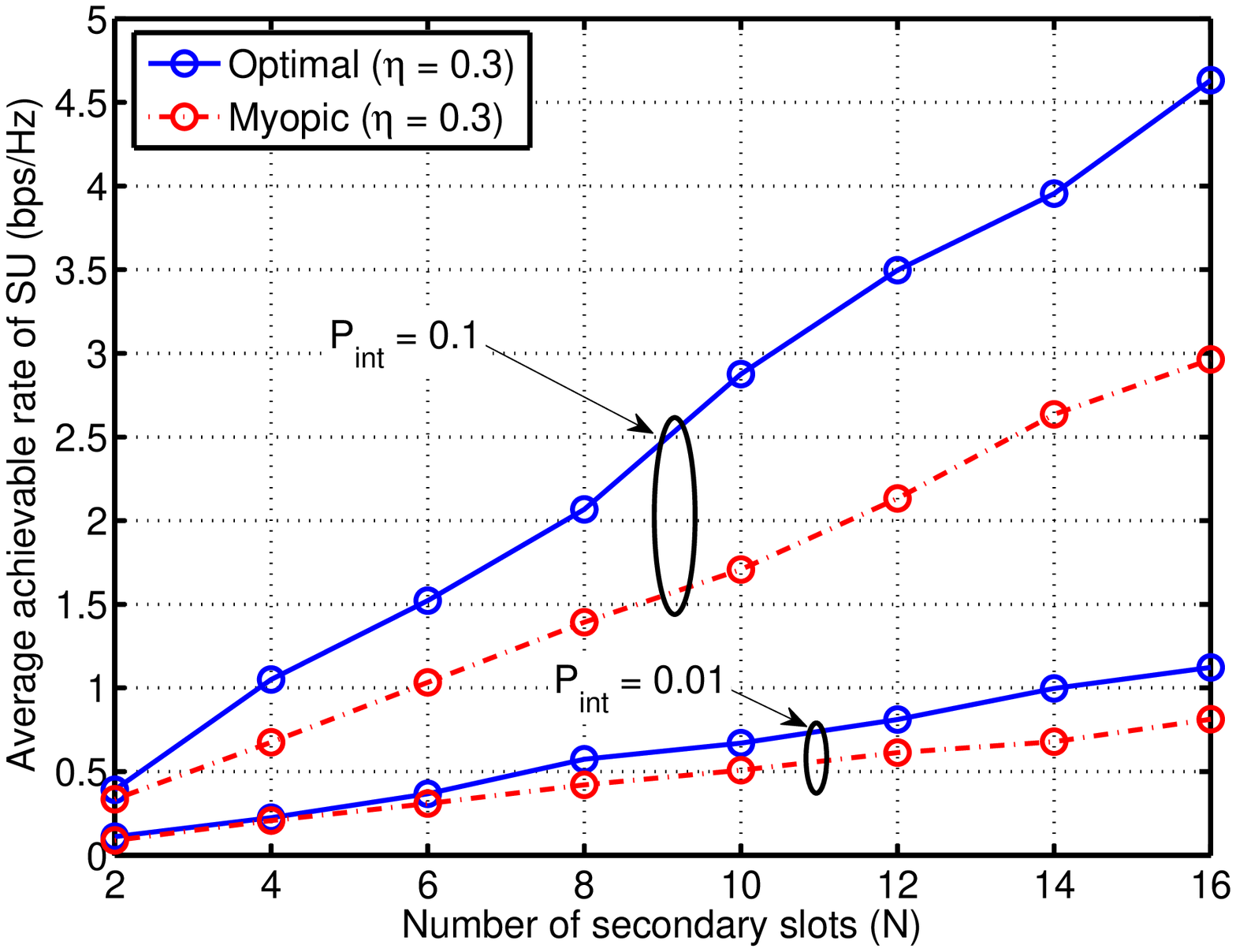}
  \caption{Average achievable sum rate of optimal and online myopic policies.}
  \label{fig:comparison}
\end{minipage}
\hspace{3mm}
\begin{minipage}{0.3\linewidth}
\centering
\centering
  \includegraphics[width=2.3in]{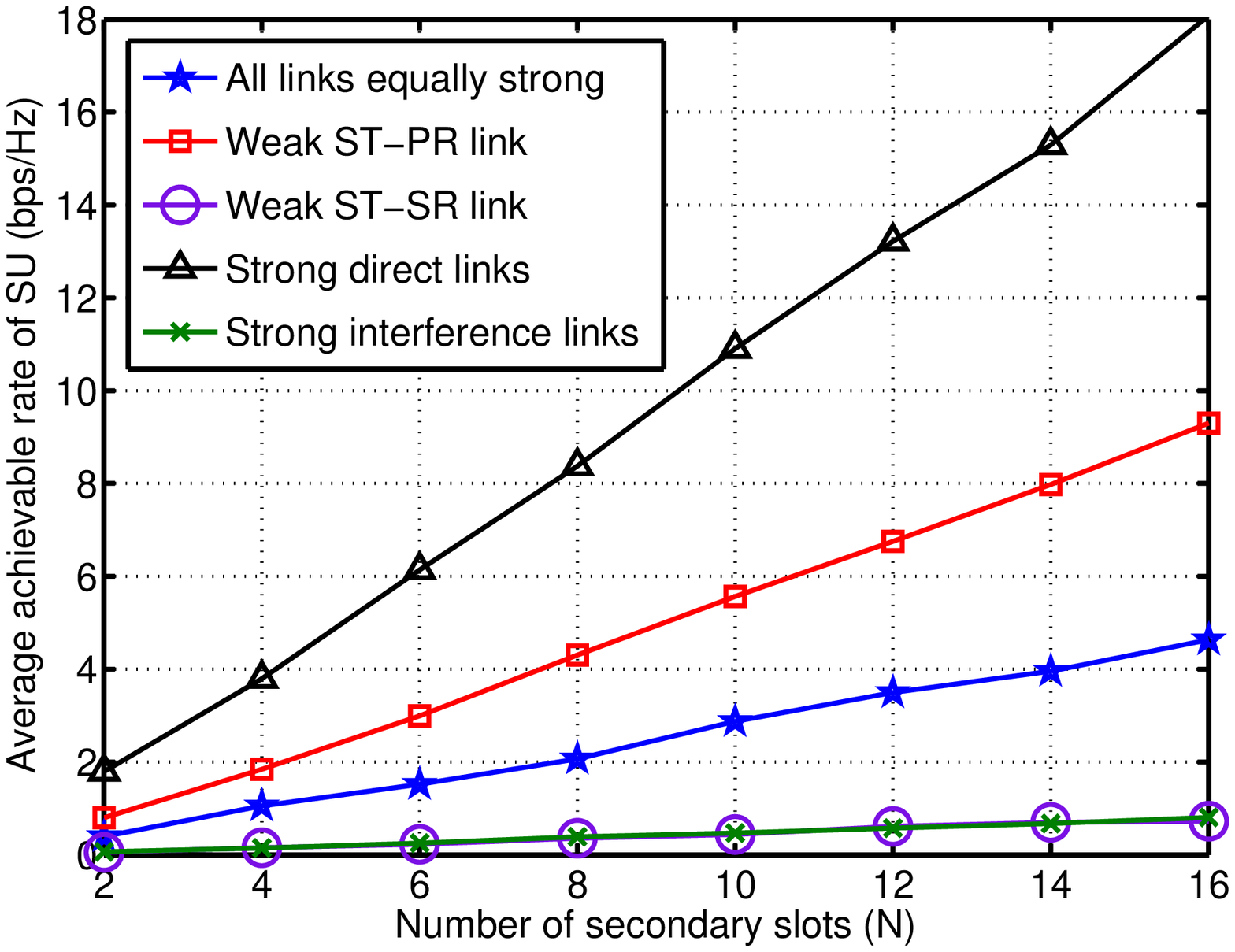}
  \caption{Average achievable rate of ST ($\vec{R}$) versus number of secondary slots ($N$) under different channel conditions ($\eta = 0.3$ and $P_{int}=0.1$).}
  \label{fig:diff_chan}
\end{minipage}
\vspace{-0.3cm}
\end{figure*}
\subsection{Optimal Solution}
The Lagrangian of the problem (\ref{eq:convex_opt_1})-(\ref{eq:convex_opt_4}) is given as:
\begin{align}
\mathcal{L}(\vec{x},\vec{y})= &-\sum_{i=1}^M\alpha_i\log_2\left( 1+ \frac{h_{ss}^iE_s^i}{\alpha_i(\sigma_s^2+h_{ps}^iP_p)}\right)\nonumber \\
&+\sum_{i=1}^M \lambda_i \left[\sum_{j=1}^iE_s^j-\sum_{j=1}^i(1-\alpha_j)\eta P_p \right]\nonumber \\
&+\sum_{i=1}^M\gamma_i[h_{sp}^iE_s^i-P_{int}\alpha_i]+\sum_{i=1}^M\mu_i(\alpha_i-1) \label{eq:lagrangian}
\end{align}
where $\vec{\pmb{\lambda}},\vec{\pmb{\gamma}}$ and $\vec{\pmb{\mu}}$ are the dual variables for the constraints (\ref{eq:convex_opt_2}), (\ref{eq:convex_opt_3}) and (\ref{eq:convex_opt_4}) respectively and $\vec{x}\in\mathcal{X}=\{\vec{E}_s,\vec{\pmb{\alpha}}\}$, and $\vec{y}\in\mathcal{Y}=\{\vec{\pmb{\lambda}},\vec{\pmb{\gamma}},\vec{\pmb{\mu}}\}$. The
Karush-Kuhn-Tucker (KKT) stationarity conditions are:
\begin{align}
\frac{\frac{\theta_iE_s^{i*}}{\alpha_i^*}}{\ln2\cdot\left(1+\frac{\theta_iE_s^{i*}}{\alpha_i^*}\right)}-\log_2\left(1+\frac{\theta_iE_s^{i*}}{\alpha_i^*}\right)+\sum_{j=i}^M\eta P_p\lambda_j^*\nonumber\\
-P_{int}\gamma_i^*+\mu_i^*&=0 \label{eq:kkt_1}\\
-\frac{\theta_i}{\ln2\left(1+\frac{\theta_iE_s^{i*}}{\alpha_i^*}\right)}+\sum_{j=i}^M\lambda_j^*+\gamma_i^*h_{sp}^i&=0\label{eq:kkt_2}
\end{align}
where $\theta_i=\frac{h_{ss}^i}{\sigma_s^2+h_{ps}^iP_p}$. And the complementary slackness conditions are
\begin{align}
\lambda_i^*\left[\sum_{j=1}^iE_s^{j*}-\sum_{j=1}^i(1-\alpha_j^*)\eta P_p\right]&=0\\
\gamma_i^*[h_{sp}^iE_s^{i*}-\alpha_i^*P_{int}]&=0\\
\mu_i^*[\alpha_i^*-1]&=0
\end{align}
for all $i=1,\ldots,M$. To make computation more tractable, we neglect the dual variables associated with non-negativity constraints of $\vec{E}_s$ and $\vec{\pmb{\alpha}}$. We can include these constraints later by projecting the obtained results onto positive orthant. From (\ref{eq:kkt_1}) and (\ref{eq:kkt_2}) we obtain the optimal solution given in eq. (\ref{eq:opt_sol}).

Since the objective function is strictly concave, the optimal $\vec{E}_s$ and $\vec{\pmb{\alpha}}$ are unique. Since $P_s^{i*}=\frac{E_s^{i*}}{\alpha_i^*}$, the optimal transmit power of ST in $i$th slot is given as:
\begin{align}
P_s^{i*}=\left\{
\begin{array}{cl}
\left[\frac{1}{\ln2(\sum_{j=i}^M\lambda_j^*+\gamma_i^*h_{sp}^i)}-\frac{1}{\theta_i}\right]^+, & 0<\alpha_i\leq1\\
0, & \alpha_i=0
\end{array}
\right.\label{eq:opt_pwr}
\end{align}
\subsection{Online Myopic Policy}
\indent To compare our proposed scheme, we consider an online myopic policy as in \cite{underlay_3}, in which only instantaneous channel gains are known at ST. However, unlike \cite{underlay_3}, we consider rather more strict constraint for interference at PR. In our system model, we consider that in each slot, interference at PR must be less than the threshold $P_{int}$, whereas in \cite{underlay_3}, authors considered that the probability of interference being greater than threshold should be arbitrarily small. As in \cite{underlay_3}, in each slot the ST uses all the harvested energy in the same slot. Thus, in $i$th slot, the consumed energy will be $E_s^i=(1-\alpha_i)\eta P_p$. The optimal time sharing parameter $\vec{\pmb{\alpha}}$, which maximizes the short term secondary sum rate, can be obtained by solving the following optimization problem:
\vspace{-0.1cm}
\begin{subequations}
\begin{align}
\max_{\vec{\pmb{\alpha}}}\quad & \sum_{i=1}^M \alpha_i \log_2\left(1+\frac{1-\alpha_i}{\alpha_i}\frac{h_{ss}^i\eta P_p}{\sigma_s^2+h_{ps}^iP_p}\right) \label{eq:subopt_1}\\
\text{s.t.}\quad & h_{sp}^i(1-\alpha_i)\eta P_p\leq\alpha_i P_{int},\quad i=1,\ldots,M.\label{eq:subopt_2}\\
& \vec{0}\prec\vec{\pmb{\alpha}}\prec\vec{1} \label{eq:subopt_3}
\end{align}
\end{subequations}
\vspace{-0.1cm}
This is a convex optimization problem as optimization function is negative of sum of relative entropies, $D(p_i||q_i)$ where $p_i=\alpha_i$ and $q_i=\alpha_i+(1-\alpha_i)\frac{h_{ss}^i\eta P_p}{\sigma_s^2+h_{ps}^iP_p}$, and the constraints are linear inequalities. Since ST uses all the harvested energy in the same slot, the optimization problem becomes independent among slots and we can decouple the problem into $M$ parallel subproblems. By rearranging the terms, constraint (\ref{eq:subopt_2}) can be written as $\alpha_i\leq\frac{h_{sp}^i\eta P_p}{P_{int}+h_{sp}^i\eta P_p}$, which is $\leq1$ unless $P_{int}=0$. Therefore, the constraints (\ref{eq:subopt_2}) and (\ref{eq:subopt_3}) can be merged together into a single constraint $0<\alpha_i\leq\frac{h_{sp}^i\eta P_p}{P_{int}+h_{sp}^i\eta P_p}$. In $i$th slot, we solve the following convex optimization problem:
\begin{subequations}
\begin{align}
\max_{\alpha_i} \quad & \alpha_i \log_2\left(1+\frac{1-\alpha_i}{\alpha_i}\frac{h_{ss}^i\eta P_p}{\sigma_s^2+h_{ps}^iP_p}\right) \label{eq:myopic_1}\\
\text{s.t.} \quad & 0<\alpha_i\leq\frac{h_{sp}^i\eta P_p}{P_{int}+h_{sp}^i\eta P_p}
\end{align}
\end{subequations}
The Lagrangian of the problem stated above is given as:
\begin{align*}
\mathcal{L}(\alpha_i,\lambda)=&-\alpha_i \log_2\left(1+\frac{1-\alpha_i}{\alpha_i}\zeta_i\right)+\lambda\left(\alpha_i-\Psi_i\right)
\end{align*}
where $\zeta_i=\frac{h_{ss}^i\eta P_p}{\sigma_s^2+h_{ps}^iP_p}$ and $\Psi_i = \frac{h_{sp}^i\eta P_p}{P_{int}+h_{sp}^i\eta P_p}$. Using the KKT optimality conditions, the closed form of optimal time sharing parameter $\alpha_i$ is given as \cite{myopic_solution}:
\begin{align}
\alpha_i^*=\max\left\{\frac{\zeta_i}{\zeta_i+z_i^*-1},\frac{h_{sp}^i\eta P_p}{P_{int}+h_{sp}^i\eta P_p}\right\}\label{eq:optimal_myopic}
\end{align}
where $z_i^*>1$ is the unique solution of following equation:
\begin{align*}
z_i\ln z_i-z_i-\zeta_i +1=0
\end{align*}
\section{Results and Discussions}
For simulation purpose, we initially assume that all channel links are i.i.d. Rayleigh distributed with variance $\sigma_{pp}^2=\sigma_{sp}^2=\sigma_{ps}^2=\sigma_{ss}^2=1$. The primary user transmits with power $P_p=2$ Watt in all the slots. The energy harvesting link is assumed to be static throughout the transmission and has an attenuation factor $\eta$. Therefore, the energy harvested by ST in each slot is $E_h^i=(1-\alpha_i)\eta P_p$ Joule. The noise at both the receivers is assumed to be zero mean additive white Gaussian with equal variances $\sigma_p^2=\sigma_s^2=0.1$. In addition, we assume that there is no initial energy available in the battery. The results for Figs. \ref{fig:opt_plot} and \ref{fig:comparison} are obtained for the system settings as stated before, whereas for Fig. \ref{fig:diff_chan}, we change the channel variances to analyze the effects of different channel conditions.
\subsection{Effects of harvesting efficiency and interference constraint}
Fig. \ref{fig:opt_plot} shows the average achievable sum rate of ST ($\vec{R}_{sum}$) averaged over different channel realizations versus number of secondary slots ($N$) for different values of interference threshold $P_{int}$ and energy harvesting efficiency $\eta$ under the optimal offline policy. The average achievable sum rate is monotonically increasing with the number of secondary slots as expected. In addition, higher $\eta$ results in higher achievable rate as with higher $\eta$, secondary harvests more energy which allows it to transmit with higher power. Also, it can be observed that as the interference constraint at PR becomes stricter (i.e. $P_{int}$ decreases), the average achievable rate of ST reduces as the interference constraint puts an upper bound on the transmit power $\vec{P}_s$. In addition, with stricter interference constraint, the rate at which $\vec{R}_{sum}$ increases w.r.t. $N$, also reduces. 
\subsection{Comparison of optimal offline and myopic policies}
Fig. \ref{fig:comparison} shows the comparison of average achievable sum rate obtained using optimal offline and online myopic policies. The plot is obtained for a fixed energy harvesting efficiency $\eta=0.3$. The figure clearly shows that the myopic policy with causal CSI performs worse than the optimal policy with non-causal CSI. This is expected as myopic policy does not have any knowledge of future channel gains and the ST needs to consume all the harvested energy in the same slot. Hence, the ST harvests only that much energy in a slot which it can consume in the same and tries to maximize the achievable rate of that slot only. It is observed that this difference in performances reduces as interference constraint becomes stricter. This is because in offline policy, with stricter constraint, the ST harvests very less amount of energy and transmits with very low power and saves very little amount of energy for future use.
\subsection{Effects of different channel conditions}
Fig. \ref{fig:diff_chan} shows the average achievable sum rate of SU under different channel conditions. The results for weak ST-PR and weak PT-SR links are obtained by choosing $\sigma_{sp}^2=0.1$ and $\sigma_{ps}^2=0.1$ for ST-PR and PT-SR links respectively, whereas the variance of other channel links are assumed to be 1. From the figure it is observed that as ST-PR link becomes weak, ST's achievable rate increases significantly. This is because weak ST-PR link is as good as loose interference constraint on PR, which allows ST to transmit with higher power yielding higher sum rate. Also, weak ST-SR link decreases the sum rate as due to interference constraint, the ST transmit power is upper bounded by $\frac{P_{int}}{h_{sp}^i}$ and small channel coefficients of ST-SR link cause the sum-rate to decrease. The effects of strong direct links ($\sigma_{pp}^2=\sigma_{ss}^2=1$ and $\sigma_{ps}^2=\sigma_{sp}^2=0.1$) and strong interference links ($\sigma_{pp}^2=\sigma_{ss}^2=0.1$ and $\sigma_{ps}^2=\sigma_{sp}^2=1$) are also shown. As expected, the system performance is much better in the case of strong direct links as compared to the case of strong interference links.
\section{Conclusion}
\indent We considered an energy harvesting underlay cognitive radio system operating in slotted fashion. We considered a case where PU is equipped with a reliable power source whereas the SU harvests energy from PU's transmission. Each secondary slot is divided into two phases: EH and IT phase. Assuming complete CSI to be available non-causally at ST, we obtained an optimal offline time sharing and transmit power policy which maximizes the sum achievable rate of SU. Also, we obtained closed form expressions for time sharing parameter $\vec{\pmb{\alpha}}$ and consumed energy $\vec{E}_s$ so that optimal values of these parameters can be found iteratively. In addition, considering only causal CSI to be available at ST, we obtained an online myopic policy maximizing the achievable sum rate of SU. We compared the performance of both the policies and it is observed that the offline policy outperforms the myopic policy. But since the myopic policy needs only causal CSI at ST, it is more practical than the offline policy.

\footnotesize
\bibliographystyle{IEEEtr} 
\bibliography{references}

\end{document}